\documentclass[a4paper,UKenglish]{lipics-v2018}

\usepackage{microtype}%if unwanted, comment out or use option "draft"
\usepackage{float}
\floatstyle{boxed} 
\restylefloat{figure}
\usepackage{algpseudocode}
\usepackage{algorithmicx}
\usepackage{tikz}

\title{Simulating Random Walks on Graphs in the Streaming Model}

\author{Ce Jin}{Institute for Interdisciplinary Information Sciences, Tsinghua University, Beijing, China}{jinc16@mails.tsinghua.edu.cn}{}{}

\authorrunning{C.\,Jin}

\Copyright{Ce Jin}

\subjclass{
\ccsdesc[100]{Theory of computation~Streaming models}
}% mandatory: Please choose ACM 2012 classifications from https://www.acm.org/publications/class-2012 or https://dl.acm.org/ccs/ccs_flat.cfm . E.g., cite as "General and reference $\rightarrow$ General literature" or \ccsdesc[100]{General and reference~General literature}. 

\keywords{streaming models, random walks, sampling}%mandatory

\category{}%optional, e.g. invited paper

\relatedversion{}%optional, e.g. full version hosted on arXiv, HAL, or other respository/website

\supplement{}%optional, e.g. related research data, source code, ... hosted on a repository like zenodo, figshare, GitHub, ...

\funding{}%optional, to capture a funding statement, which applies to all authors. Please enter author specific funding statements as fifth argument of the \author macro.

\acknowledgements{I would like to thank Professor Jelani Nelson for introducing this problem to me, advising this project, and giving many helpful comments on my writeup.}%optional

\EventEditors{Avrim Blum}

\EventNoEds{1}

\EventLongTitle{10th Innovations in Theoretical Computer Science Conference (ITCS 2019)}

\EventShortTitle{ITCS 2019}

\EventAcronym{ITCS}

\EventYear{2019}

\EventDate{January 10--12, 2019}

\EventLocation{San Diego, California, USA}

\EventLogo{}

\SeriesVolume{124}

\ArticleNo{44} % New number (=<article-no>) goes here!

\nolinenumbers %uncomment to disable line numbering
\hideLIPIcs  %uncomment to remove references to LIPIcs series (logo, DOI, ...), e.g. when preparing a pre-final version to be uploaded to arXiv or another public repository
\newcommand{\E}{\mathbb{E}}
\newcommand{\R}{\mathbb{R}}
\newcommand{\1}{\mathbf{1}}

\newcommand{\eps}{\varepsilon}
\renewcommand{\P}{\mathbb{P}}

\begin{document}

\maketitle

\begin{abstract}
We study the problem of approximately simulating a $t$-step random walk on a graph where the input edges come from a single-pass stream.
The straightforward algorithm using reservoir sampling needs $O(nt)$ words of memory. We show that this space complexity is near-optimal for directed graphs. 
For undirected graphs, we prove an $\Omega(n\sqrt{t})$-bit space lower bound, and give a near-optimal algorithm using $O(n\sqrt{t})$
words of space with $2^{-\Omega(\sqrt{t})}$ simulation error  (defined as the $\ell_1$-distance between the output distribution of the simulation algorithm and the distribution of perfect random walks). 
We also discuss extending the algorithms to the turnstile model, where both insertion and deletion of edges can appear in the input stream.
 \end{abstract}

\section{Introduction}
Graphs of massive size are used for modeling complex systems that emerge in many different fields of study. Challenges arise when computing with massive graphs under memory constraints. In recent years, graph streaming has become  an important model for computation on massive graphs. Many space-efficient streaming algorithms   have been designed for solving classical graph problems, including connectivity \cite{ahn2012analyzing}, bipartiteness \cite{ahn2012analyzing}, minimum spanning tree \cite{ahn2012analyzing}, matching \cite{epstein2011improved, kapralov2013better,ahn2013linear}, spectral sparsifiers \cite{kelner2013spectral, kapralov2017single}, etc. We will define the streaming model in Section~\ref{introstreaming}.

Random walks on graphs are stochastic processes that have many applications, such as connectivity testing \cite{reingold2008undirected}, clustering \cite{spielman2013local,andersen2007using, andersen2009finding, charikar2003better}, sampling \cite{jerrum1986random} and approximate counting \cite{jerrum1989approximating}.
Since random walks are a powerful tool in algorithm design, it is interesting to study them in the streaming setting.
A natural problem is to find the space complexity of simulating random walks in  graph streams. 
Das~ Sarma~et~al.~\cite{sarma2011estimating}   gave a multi-pass streaming algorithm that simulates a $t$-step random walk on a directed graph using $O(\sqrt{t})$ passes and only $O(n)$ space. By further extending this algorithm and combining with other ideas, they obtained space-efficient algorithms for estimating PageRank on graph streams.
However, their techniques crucially rely on reading  multiple passes of the input stream. 

In this paper, we study the problem of simulating random walks in the \textit{one-pass} streaming model. We show space lower bounds for both directed and undirected versions of the problem, and present algorithms that  nearly match with the lower bounds. We summarize our results in Section~\ref{secres}.

\subsection{One-pass streaming model}
\label{introstreaming}
Let $G=(V,E)$ be a graph with $n$ vertices.
In the insertion-only model, the input graph $G$ is defined by a stream of edges $(e_1,\dots,e_m)$ seen in arbitrary order, where each edge $e_i$ is specified by its two endpoints $u_i,v_i \in V$.
An algorithm must process the edges of $G$ in the order that they appear in the input stream.  
The edges can be directed or undirected, depending on the problem setting.
Sometimes we allow multiple edges in the graph, where the multiplicity of an edge equals its number of occurrences in the input stream.

In the turnstile model, we allow both insertion and deletion of edges. The input is a stream of updates $((e_1,\Delta_1),(e_2,\Delta_2),\dots)$, where $e_i$ encodes an edge and $\Delta_i \in \{1,-1\}$. The multiplicity of edge $e$ is $f(e)=\sum_{e_i=e}\Delta_i$. We assume $f(e) \ge 0$ always holds for every edge $e$.
\subsection{Random walks}
Let $f(u,v)$ denote the multiplicity of edge $(u,v)$. 
The degree of $u$ is defined by $d(u) =\sum_{v\in V}f(u,v)$.
A $t$-step random walk starting from a vertex $s \in V$ is a random sequence of vertices $v_0,v_1,\dots,v_t$ where $v_0=s$ and $v_i$ is a vertex uniformly randomly chosen from the vertices that $v_{i-1}$ connects to, i.e., $\P[v_i = v|v_{i-1} = u] = f(u,v)/d(u)$.
 Let $\mathcal{RW}_{s,t} : V^{t+1} \to [0,1]$ denote the distribution of $t$-step random walks starting from $s$, defined by\footnote{For a statement $p$, define $\1[p] = 1$ if $p$ is true, and $\1[p]=0$ if $p$ is false. }
\begin{equation}
 \mathcal{RW}_{s,t}(v_0,\dots,v_t) = \1[v_0=s] \prod_{i=0}^{t-1} \frac{f(v_i,v_{i+1})}{d(v_i)}.
\end{equation}
For two distributions $P,Q$, we denote by $|P-Q|_1$ their $\ell_1$ distance. 
 We say that a randomized algorithm can simulate a $t$-step random walk starting from $v_0$ within error $\eps$, if the distribution  $\P_w$ of its output $w\in V^{t+1}$  satisfies $|\P_w - \mathcal{RW}_{v_0,t}|_1 \le \eps$. We say the random walk simulation is perfect if $\eps=0$.

We study the problem of simulating a $t$-step random walk within error $\eps$ in the streaming model using small space.
We assume the length $t$ is specified at the beginning.
Then the algorithm reads the input stream.
When a query with parameter $v_0$ comes, the algorithm should simulate and output a $t$-step random walk starting from vertex $v_0$.

It is without loss of generality to assume that the input graph has no self-loops.  If we can simulate a random walk on the graph with self-loops removed, we can then turn it into a random walk of the original graph by simply inserting self-loops after $u$ with probability $d_{\text{self}}(u)/d(u)$. The values $d_{\text{self}}(u),d(u)$ can be easily maintained by a streaming algorithm using $O(n)$ words.

The random walk is not well-defined when it starts from a vertex $u$ with $d(u)=0$. For undirected graphs, this can only happen at the beginning of the random walk, and we simply let our algorithm return  \textsc{Fail} if $d(v_0)=0$. For directed graphs, one way to fix this is to continue the random walk from $v_0$, by adding an edge $(u,v_0)$ for every vertex $u$ with $d(u)=0$.
We will not deal with $d(u)=0$ in the following discussion.

\subsection{Our results}
\label{secres}

We will use $\log x = \log_2 x$ throughout this paper.

The following two theorems give space lower bounds on directed and undirected versions of the problem. 
Note that the lower bounds hold even for simple graphs\footnote{A \textit{simple} graph is a graph with no multiple edges.}.
\begin{theorem}
For $t \le n/2$, simulating a $t$-step random walk on a simple directed graph in the insertion-only model within error $\eps = \frac{1}{3}$ requires $\Omega(nt \log (n/t))$ bits of memory.
\end{theorem}
\begin{theorem}
For $t = O(n^2)$, simulating a $t$-step random walk on a simple undirected graph in the insertion-only model within error $\eps=\frac{1}{3}$ requires $\Omega(n\sqrt{t})$ bits of memory.
\end{theorem}
Theorem~\ref{maindirected} and Theorem~\ref{mainundirected} give near optimal space upper bounds  for the problem in the insertion-only streaming model. 
\begin{theorem}
\label{maindirected}
 We can simulate a $t$-step random walk on a directed graph in the insertion-only model perfectly using $O(nt)$ words\footnote{A \textit{word} has $\Theta(\log \max\{n,m\})$ bits.} of memory. For simple directed graphs, the memory can be reduced to $O(nt\log (n/t))$ bits, assuming $t \le n/2$.
\end{theorem}
\begin{theorem}
\label{mainundirected}
We can simulate a $t$-step random walk on an undirected graph in the insertion-only model within error $\eps$ using $O\left ( n \sqrt{t} \cdot \frac{q}{\log q} \right )$ 
words of memory, where $q = 2+\frac{\log(1/\eps)}{\sqrt{t}}$.  In particular, the algorithm uses $O(n\sqrt{t})$ words of memory when $\eps =   2^{-\Theta\left (\sqrt{t}\right )}$.
\end{theorem}
Our algorithms also extend to the turnstile model. 
\begin{theorem}
  We can simulate a $t$-step random walk on a directed graph in the turnstile model within error $\eps$ using $O(n(t+ \log\frac{1}{\eps})\log^2 \max\{n,1/\eps\})$ bits of memory.
\end{theorem}
\begin{theorem}
  We can simulate a $t$-step random walk on an undirected graph in the turnstile model within error $\eps$ using $O(n(\sqrt{t}+ \log\frac{1}{\eps})\log^2 \max\{n,1/\eps\})$ bits of memory.
\end{theorem}
\section{Directed graphs in the insertion-only model}
The simplest algorithm uses $O(n^2)$ words of space  (or only $O(n^2)$ \textit{bits}, if we assume the graph is simple) to store the adjacency matrix of the graph. 
When $t\ll n$, a better solution is to use reservoir sampling.
  \begin{lemma}[Reservoir sampling]
  Given a stream of $n$ items as input, we can uniformly sample $m$ of them without replacement using $O(m)$ words of memory.
  \end{lemma}
  We can also sample $m$ items from the stream \textit{with} replacement in $O(m)$ words of memory using $m$ independent reservoir samplers each with capacity 1.
\begin{theorem}
\label{algdirected}
 We can simulate a $t$-step random walk on a directed graph in the insertion-only model perfectly using $O(nt)$ words of memory.
\end{theorem}
\begin{proof}
 For each vertex $u\in V$, we sample $t$ edges $e_{u,1},\dots,e_{u,t}$ outgoing from $u$ with replacement. Then we perform a random walk using these edges. When $u$ is visited for the $i$-th time ($i \le t$), we go along edge $e_{u,i}$.  
\end{proof}

By treating an undirected edge as two opposite directed edges, we can achieve the same space complexity in undirected graphs.

Now we show a space lower bound for the problem. We will use a standard result from communication complexity.

\begin{definition}
In the \textsc{Index} problem, Alice has an  $n$-bit vector $X \in \{0,1\}^n$ and Bob has an index $i \in [n]$.
Alice sends a message to Bob, and then Bob should output the bit $X_i$.
\end{definition}

\begin{lemma}[\cite{miltersen1998data}]
\label{index}
For any constant $1/2<c\le 1$, solving the \textsc{Index} problem with success probability $c$ requires sending $\Omega(n)$ bits.
\end{lemma}
\begin{theorem}
For $t \le n/2$, simulating a $t$-step random walk on a simple directed graph in the insertion-only model within error $\eps = \frac{1}{3}$ requires $\Omega(nt \log (n/t))$ bits of memory.
\end{theorem}

\begin{proof}
We prove by showing a reduction from the \textsc{Index} problem.  Before the protocol starts, Alice and Bob agree on a family $\mathcal{F}$ of $t$-subsets of $[n]$ \footnote{Define $[n] = \{1,2,\dots,n\}$. A $t$-subset is a subset of size $t$.} such that the condition $|S \cap S'| < t/2$ is satisfied for every $S,S'\in \mathcal{F}, S\neq S'$. 
For two independent uniform random $t$-subsets $S,S'\subseteq [n]$, let $p=\P[|S\cap S'|\ge t/2] \le \binom{t}{t/2} (\frac{t}{n})^{t/2}<(\frac{4t}{n})^{t/2}$.
By union bound over all pairs of subsets, a randomly generated family $\mathcal{F}$ satisfies the condition with probability at least $1- \binom{|\mathcal{F}|}{2} p$, which is positive when $|\mathcal{F}| = \lceil \sqrt{1/p} \rceil \ge (\frac{n}{4t})^{t/4}$.
So we can choose such family $\mathcal{F}$ with $\log |\mathcal{F}| = \Omega(t\log(n/t))$.

Assume $|\mathcal{F}|$ is a power of two. Alice encodes $n  \log | \mathcal{F}|$ bits  as follows. Let $G$ be a directed graph with vertex set $\{v_0,v_1,\dots,v_{2n}\}$. For each vertex $u \in \{v_{n+1},v_{n+2},\dots,v_{2n}\}$, Alice chooses a set $S_u \in \mathcal{F}$, and inserts an edge $(u,v_i)$ for every $i\in S_u$. 

Suppose Bob wants to query $S_u$.  He adds an edge $(v,u)$ for every $v \in \{v_0,v_1,v_2,\dots,v_n\}$, and then simulates a random walk starting from $v_0$.
The random walk visits $u$ every two steps, and it next visits $v_i$ for some random $i\in S_u$.
At least $t/2$ different elements from $S_u$ can be seen in $2t$ samples with probability at least $1-\binom{t}{t/2} (\frac{1}{2})^{2t} \ge 1-2^{-t}$, so $S_u$ can be uniquely determined by an $O(t)$-step random walk (simulated within error $\eps$) with probability $1-2^{-t}- \frac{\eps}{2}>\frac{1}{2}$.
By Lemma~\ref{index}, the space usage for simulating the $O(t)$-step random walk is at least $\Omega(n \log |\mathcal{F}|) = \Omega(nt \log(n/t))$ bits. The theorem is proved by scaling down $n$ and $t$ by a constant factor.
\end{proof}

For simple graphs, we can achieve an upper bound of $O(nt\log(n/t))$ bits. 
\begin{theorem}
For $t\le n/2$, we can simulate a $t$-step random walk on a simple directed graph in the insertion-only model perfectly using $O(nt\log (n/t))$ bits of memory.
\end{theorem}
\begin{proof}
 For every $u\in V$, we run a reservoir sampler with capacity $t$, which samples (at most) $t$ edges from  $u$'s outgoing edges \textit{without} replacement. 
 After reading the entire input stream, we begin simulating the random walk.
When $u$ is visited during the simulation, in the next step we choose at random an outgoing edge used before with probability $d_{\text{used}}(u)/d(u)$, or an unused  edge from the reservoir sampler  with probability $1-d_{\text{used}}(u)/d(u)$, where $d_{\text{used}}(u)$ is the number of edges in $u$'s  sampler that are previously used in the simulation.
We  maintain a $t$-bit vector to keep track of these used samples.
 
 The number of different possible states of a sampler is at most $\sum_{0\le i \le t}\binom{n}{i} \le (t+1)(\frac{en}{t})^t $, so it can be encoded using $\Big \lceil \log \Big ((t+1)(\frac{en}{t})^t\Big ) \Big \rceil = O(t \log(n/t))$ bits. The total space is $O(nt \log(n/t))$ bits.
\end{proof}

\section{Undirected graphs in the insertion-only model}
\subsection{A space lower bound}
\begin{theorem}
\label{thmlbundir}
For $t = O(n^2)$, simulating a $t$-step random walk on a simple undirected graph in the insertion-only model within error $\eps=\frac{1}{3}$ requires $\Omega(n\sqrt{t})$ bits of memory.
\end{theorem}
\begin{proof}
Again we show a reduction from the \textsc{Index} problem.

Alice encodes $\Omega(n\sqrt{t})$ bits as follows. Let $G$ be an undirected graph with vertex set $V_0 \cup V_1 \cup \dots \cup V_{n/\sqrt{t}}$, where each  $V_j$ has size $2\sqrt{t}$, and the starting vertex $v_0 \in V_0$. For each $j\ge 1$, $V_j$ is divided into two subsets $A_j,B_j$ with size $\sqrt{t}$ each, and Alice encodes $|A_j|\times |B_j| = t$ bits by inserting a subset of edges from $\{ (u,v):  u\in A_j, v\in B_j\}$. In total she encodes $t\cdot n/\sqrt{t} = n\sqrt{t}$ bits.

Suppose Bob wants to query some bit, i.e., he wants to see whether $a$ and $b$ are connected by an edge. Assume $(a,b) \in A_j \times B_j$. He adds an edge $(u,v)$ for every $u\in A_j$ and every $v \in V_0$ (see Figure~\ref{fig:lbundir}).   A perfect random walk starting from $v_0 \in V_0$ will be inside the bipartite subgraph $(A_j, B_j\cup V_0)$.  Suppose the current vertex of the perfect random walk is $v_i \in A_j$. If $a,b$ are connected by an edge, then
\begin{align*}
&\P[(v_{i+2},v_{i+3})=(a,b) \mid v_i] \\
\ge {}& \P[v_{i+1} \in V_0 \mid v_i]\,\P[v_{i+2} = a \mid v_{i+1}\in V_0] \,\P[v_{i+3} = b \mid v_{i+2}=a]\\
\ge {}& \frac{|V_0|}{|V_0|+|B_j|} \cdot \frac{1}{|A_j|} \cdot \frac{1}{|V_0|+|B_j|}\\
\ge {}& \frac{2}{9t},
\end{align*}
so in every four steps the edge $(a,b)$ is passed with probability $\Omega(\frac{1}{t})$.  Then a $O(t)$-step perfect random walk will pass the edge $(a,b)$ with probability $0.9$.  Hence Bob can know whether the edge $(a,b)$ exists by looking at the random walk (simulated within error $\eps$) with success probability $0.9-\frac{\eps}{2}>1/2$.
By Lemma~\ref{index}, the space usage for simulating the $O(t)$-step random walk is at least $\Omega(n \sqrt{t})$ bits. The theorem is proved by scaling down $n$ and $t$ by a constant factor.
\begin{figure}
\begin{center}
\begin{tikzpicture}
\tikzstyle{vertex}=[circle,draw,text=black]
\tikzstyle{txt}=[text=black]
\tikzstyle{edge} = [draw,black,line width = 1.4pt,-]
\draw[rounded corners, draw=white, fill = black!15] (3,5) rectangle (3.9,6.8);
\draw[rounded corners, draw=white, fill = black!15] (5,5) rectangle (5.9,6.8);
\draw[rounded corners, draw=white, fill = black!15] (7,5) rectangle (7.9,6.8);
\draw[rounded corners, draw=white, fill = black!15] (5,1.2) rectangle (5.9,4.4);
\fill [fill=black!80] (3.45,5.4) circle(0.1) node[left] {$a$};
\fill [fill=black!80] (3.45,6.4) circle(0.1) node[left] {$ $};
\fill [fill=black!80] (5.45,6.4) circle(0.1) node[right] {$b$};
\fill [fill=black!80] (5.45,5.4) circle(0.1) node[right] {$ $};
\node[txt] (2) at (3.45,7.2) {$A_j$};
\node[txt] (1) at (5.45,7.2) {$B_j$};
\node[txt] (1) at (7.45,7.2) {$A_{j+1}$};
\node[txt] (1) at (5.45,0.8) {$V_0$};
\node[txt] (1) at (5.45,6.0) {$\vdots$};
\node[txt] (1) at (9.45,6.0) {$\cdots$};
\node[txt] (1) at (5.45,2.0) {$\vdots$};
\node[txt] (1) at (3.45,6.0) {$\vdots$};
\fill [fill=black!80] (5.45,3.9) circle(0.1) node[right] {};
\fill [fill=black!80] (5.45,3.3) circle(0.1) node[right] {};
\fill [fill=black!80] (5.45,2.7) circle(0.1) node[right] {$v_0$};

\path[draw=black,line width = 0.6pt](3.45,5.4)--(5.45,6.4);
\path[draw=black,line width = 0.6pt](3.45,6.4)--(5.45,5.4);
\path[draw=black,line width = 0.6pt](3.45,5.4)--(5.45,3.9);
\path[draw=black,line width = 0.6pt](3.45,5.4)--(5.45,3.3);
\path[draw=black,line width = 0.6pt](3.45,5.4)--(5.45,2.7);
\path[draw=black,line width = 0.6pt](3.45,6.4)--(5.45,3.9);
\path[draw=black,line width = 0.6pt](3.45,6.4)--(5.45,3.3);
\path[draw=black,line width = 0.6pt](3.45,6.4)--(5.45,2.7);
\end{tikzpicture}
\end{center}
\caption{Proof of Theorem~\ref{thmlbundir}}
\label{fig:lbundir} 
\end{figure}
\end{proof}
\subsection{An algorithm for simple graphs}
Now we describe our algorithm for undirected graphs in the insertion-only model.  As a warm-up, we consider simple graphs in this section. We will deal with multi-edges in Section~\ref{secmulti}.

\paragraph*{Intuition}
We start by informally explaining the intuition of our algorithm for simple undirected graphs.

 We maintain a subset of $O(n\sqrt{t})$ edges from the input graph, and use them to simulate the random walk after reading the entire input stream. 

For a vertex $u$ with degree smaller than $\sqrt{t}$, we can afford to store all its neighboring edges in memory. For $u$ with degree greater than $\sqrt{t}$, we can only sample and store $O(\sqrt{t})$ of its neighboring edges. During the simulation, at every step we first toss a coin to decide whether the next vertex has small degree or large degree. In the latter case, we have to pick a sampled neighboring edge and walk along it. If all sampled neighboring edges have already been used, our algorithm fails. Using the large degree and the fact that edges are undirected, we can show that the failure probability is low.

\paragraph*{Description of the algorithm}
We divide the vertices into two types according to their degrees: the set of \textit{big} vertices $B = \{u \in V: d(u)\ge C+1\}$, and the set of \textit{small} vertices $S = \{u \in V: d(u)\le C\}$, where parameter $C$ is an positive integer to be determined later.

We use \textit{arc} $(u,v)$ to refer to an edge when we want to specify the direction $u\to v$. So an undirected edge $(u,v)$ corresponds to two different\footnote{We have assumed no self-loops exist, so $u \neq v$.} arcs, arc $(u,v)$ and arc $(v,u)$.

We say an arc $(u,v)$ is \textit{important} if $v \in S$, or \textit{unimportant} if $v \in B$. Denote the set of important arcs by $E_1$, and the set of unimportant arcs by $E_0$. The total number of important arcs equals $\sum_{s\in S}d(s) \le |S|C$, so it is possible to store $E_1$ in $O(nC)$ words of space. 

The set $E_0$ of unimportant arcs can be huge, so we only store a subset of $E_0$.
For every vertex $u$, we sample with replacement $C$ unimportant arcs outgoing from $u$, denoted by $a_{u,1},\dots,a_{u,C}$. 

To maintain the set $E_1$ of important arcs and the samples of unimportant arcs after every edge insertion, we need to handle the events when some small vertex becomes big. This procedure is straightforward, as described by \textsc{ProcessInput} in Figure~\ref{fig:processinput}. Since $|E_1|$ never exceeds $nC$, and each of the $n$ samplers uses $O(C)$ words of space, the overall space complexity is $O(nC)$ words.

\begin{figure}
\begin{center}
\begin{algorithmic}
\Procedure{InsertArc}{$u,v$}
\State{$d(v) \gets d(v) +1$} 
\If{$d(v)= C+1$}\Comment{$v$ changes from small to big}
\For{$x \in V$ such that $(x,v) \in E_1$} \Comment{arc $(x,v)$ becomes unimportant}
\State{$E_1 \gets E_1 \backslash \{(x,v)\}$}
  \State{Feed arc $(x,v)$ into $x$'s sampler}
\EndFor
\EndIf
\If{$d(v)\le C$} \Comment{$v\in S$}
  \State{$E_1 \gets E_1  \cup \{(u,v)\}$}
  \Else\Comment{$v\in B$}
  \State{Feed arc $(u,v)$ into $u$'s sampler}
\EndIf
\EndProcedure
\Procedure{ProcessInput}{}
\State{$E_1 \gets \emptyset$}\Comment{Set of important arcs}
\For{$u \in V$}
      \State{$d(u) \gets 0$} 
      \State{Initialize $u$'s sampler (initially empty) which maintains $a_{u,1},\dots,a_{u,C}$}
\EndFor
\For{undirected edge $(u,v)$ in the input stream} 
\State{\textsc{InsertArc}$(u,v)$}
\State{\textsc{InsertArc}$(v,u)$}
\EndFor
\EndProcedure
\end{algorithmic}
\caption{Pseudocode for processing the input stream (for simple undirected graphs)}\label{fig:processinput} 
\end{center}
\end{figure}

 We begin simulating the random walk after \textsc{ProcessInput} finishes.
 When the current vertex of the random walk is $v$, with probability $d_1(v)/d(v)$ the next step will be along an important arc, where $d_1(v)$ denotes the number of important arcs outgoing from $v$.
 In this case we simply choose a uniform random vertex from $\{u: (v,u) \in E_1\}$ as the next vertex.
 However, if the next step is along an unimportant arc, we need to choose an unused sample  $a_{v,j}$ and go along this arc.
 If at this time all $C$ samples $a_{v,j}$ are already used, then our algorithm fails (and is allowed to return an arbitrary walk).
 The pseudocode of this simulating procedure is given in Figure~\ref{fig:simulation}.

\begin{figure}
\begin{center}
\begin{algorithmic}
\Procedure{SimulateRandomWalk}{$v_0,t$}
\For{$v \in V$} \State{$c(v) \gets 0$} \Comment{counter of used samples}
\EndFor
\For{$i = 0,\dots,t-1$}
\State{$N_1 \gets \{u: (v_{i},u) \in E_1\}$}
\State{$x \gets$ uniformly random integer from $\{1,2,\dots,d(v_{i})\}$}
\If{$x \le |N_1|$}
\State{$v_{i+1} \gets $ uniformly random vertex from $N_1$}
\Else
\State{$j \gets c(v_{i}) + 1$}
\State{$c(v_{i}) \gets j$}
\If{$j>C$}
\Return{\textsc{Fail}}
\Else
\State{$v_{i +1}\gets u$, where $(v_{i},u) = a_{v_{i},j}$}
\EndIf
\EndIf
\EndFor
\Return $(v_0,\dots,v_t)$
\EndProcedure
\end{algorithmic}
\caption{Pseudocode for simulating a $t$-step random walk starting from $v_0$} \label{fig:simulation}
\end{center}
\end{figure}

In a walk $w = (v_0,\dots,v_t)$, we say vertex $u$ \textit{fails} if $\left |\{i: v_i = u\text{ and }(v_i, v_{i+1}) \in E_0 \}\right| > C$. If no vertex fails in $w$, then our algorithm will successfully return $w$ with probability $\mathcal{RW}_{v_0,t}(w)$. Otherwise our algorithm will fail after some vertex runs out of the sampled unimportant arcs.  To ensure the output distribution is $\eps$-close to $\mathcal{RW}_{v_0,t}$ in $\ell_1$ distance, it suffices to make our algorithm fail with probability at most $\eps/2$, by choosing a large enough capacity $C$. 

To bound the probability $\P[\text{at least one vertex fails} \mid v_0=s]$\footnote{If not specified, assume the probability space is over all $t$-step random walks $(v_0,\dots,v_t)$ starting from $v_0$.}, we will bound the individual failure probability of every vertex, and then use union bound.
\begin{lemma}
\label{union}
 Suppose for every $u\in V$, $\P[\text{$u$ fails}\mid \text{$v_0=u$}] \le \delta$. Then for any starting vertex $s \in V$, $\P[\text{at least one vertex fails}\mid v_0=s]\le t\delta$. 
\end{lemma}

\begin{proof}
Fix a starting vertex $s$. For any particular $u \in V$, 
\begin{align*}
 &\P[\text{$u$ fails}\mid \text{$v_0=s$}]\\
= {}&\P[\text{$u$ fails, and $\exists i\le t-1, v_i=u$}\mid \text{$v_0=s$}]\\
 ={}&\P[\exists i\le t-1, v_i=u\mid \text{$v_0= s$}] \, \P[\text{$u$ fails}\mid \text{$v_0=s$, and $\exists i\le t-1, v_i=u$}]\\
 \le{}&\P[\text{$\exists i\le t-1, v_i=u$}\mid \text{$v_0=s$}]\, \P[\text{$u$ fails}\mid \text{$v_0=u$}]\\
 \le {}&  \P[\text{$\exists i\le t-1, v_i=u$}\mid \text{$v_0=s$}]\cdot\delta .
\end{align*}
By union bound,
\begin{align*}
  &\P[\text{at least one vertex fails}\mid v_0=s]\\
  \le {}& \sum_{u\in V} \P[\text{$u$ fails}\mid \text{$v_0=s$}]\\
  \le {}& \sum_{u\in V}
   \P[\exists i\le t-1, v_i=u\mid \text{$v_0=s$}]\cdot \delta\\
  ={} &  \E[\text{number of distinct vertices visited in $\{v_0,\dots,v_{t-1}\}$}\mid v_0=s]\cdot \delta\\
  \le{} & t\delta. 
\end{align*}
\end{proof}

\begin{lemma}
\label{capa}
We can choose integer parameter $C = O\left (  \sqrt{t}\cdot \frac{q}{\log q} \right ) $, where $q = 2+\frac{\log(1/\delta)}{\sqrt{t}}$, so that $\P[\text{$u$ fails} \mid \text{$v_0=u$}] \le \delta $ holds for every $u\in V$.
\end{lemma}

\begin{proof}
Let $d_0(u) = |  \{v : (u,v) \in E_0\} |$. 

For any $u \in V$, 
\begin{align*}
&\P[\text{$u$ fails} \mid \text{$v_0=u$}]\\
\le {}&\P[\text{$u$ fails} \mid \text{$v_0=u, (v_0,v_1) \in E_0$}].
\end{align*}

We rewrite this probability as the sum of probabilities of possible random walks in which $u$ fails. 
Recall that $u$ fails if and only if $|\{i: v_i=u, (v_i,v_{i+1}) \in E_0\}| \ge C+1$.
In the summation over possible random walks, we only keep the shortest prefix $(v_0,\dots,v_k)$ in which $u$ fails, i.e., the last step $(v_{k-1},v_k)$ is the $(C+1)$-st time walking along an unimportant arc outgoing from $u$. We have
\begin{align}
 &\P[\text{$u$ fails} \mid \text{$v_0=u, (v_0,v_1) \in E_0$}] \nonumber \\
={}&\sum_{k\le t}\sum_{\text{walk}(v_0,\dots,v_k)}\1 \bigg [\text{$v_0=v_{k-1}=u, \,\,(v_0,v_1),(v_{k-1},v_k)\in E_0,$} \nonumber \\[-6\jot] 
& \hspace{4cm}\text{$|\{i: v_i = u,(v_i, v_{i+1})\in E_0\}|=C+1$} \bigg ]\frac{1}{d_0(u)}\prod_{i=1}^{k-1}\frac{1}{d(v_i)}\nonumber \\
={}&\sum_{k\le t}\sum_{\text{walk}(v_0,\dots,v_{k-1})}\1 \bigg[\text{$v_0=v_{k-1}=u, \,\,(v_0,v_1)\in E_0,$} \nonumber\\[-6\jot]
&\hspace{5cm}\text{$|\{i: v_i = u, (v_i,v_{i+1})\in E_0\}|=C$}\bigg ]\prod_{i=1}^{k-1}\frac{1}{d(v_i)} \label{sum}.
 \end{align}
 Let $v'_i = v_{k-1-i}$. Since the graph is undirected, the vertex sequence $(v'_0,\dots,v'_{k-1})$ (the reversal of walk $(v_0,\dots,v_{k-1})$) is also a walk starting from and ending at $u$.  So the summation~(\ref{sum}) equals 
\begin{align*}
&
\sum_{k\le t}\sum_{\text{walk}(v'_0,\dots,v'_{k-1})}\1\bigg [\text{$v'_0=v'_{k-1}=u, \,\,(v'_{k-1},v'_{k-2})\in E_0,$}\nonumber \\[-6\jot]
&\hspace{6cm}\text{$|\{i: v'_i=u, (v'_i, v'_{i-1}) \in E_0\}|=C$}\bigg ]\prod_{i=0}^{k-2}\frac{1}{d(v'_i)}\\
={}&\underset{\text{random walk $(v'_0,\dots,v'_{t-1}) $}}{\P}\Big [\,|\{i: v'_{i}=u, (v'_i,v'_{i-1}) \in E_0\}|\ge C \,\Big \vert\, v'_0=u\Big ].
\end{align*}
Recall that $(v'_i,v'_{i-1}) \in E_0$ if and only if $v'_{i-1} \in B$. For any $1\le i \le t-1$ and any fixed prefix $v'_0,\dots,v'_{i-1}$,
\begin{align}
& \P \big [v'_i=u, (v'_i,v'_{i-1}) \in E_0 \, \big \vert \, v'_0,\dots,v'_{i-1}\big ] \nonumber\\
 \le {} & \1[v'_{i-1} \in B]\cdot \frac{1}{d(v'_{i-1})} \nonumber \label{cond}\\
 < {}&\frac{1}{C}.
\end{align}

Hence the probability that $|\{1\le i \le t-1: v'_i=u, (v'_i,v'_{i-1}) \in E_0\}|\ge C$ is at most 
\begin{align*}
&\binom{t-1}{C}\left (\frac{1}{C}\right )^C\\
\le {}& \left ( \frac{e(t-1)}{C}\right )^C\left (\frac{1}{C}\right )^C\\
< {}& \left ( \frac{et}{C^2}\right )^C.
\end{align*}

We set $C = \left \lceil 4 \sqrt{t}\ q/\log q \right \rceil$, where $q = 2 + \log(1/\delta)/\sqrt{t} >2$. Notice that  $q/\log^2 q > 1/4$. Then 
\begin{align*}
 C\log \left (\frac{C^2}{et}\right ) \ge  \frac{4\sqrt{t} q}{\log q} \log \left ( \frac{16q^2}{e\log^2 q}\right ) >\frac{4\sqrt{t} q}{\log q} \log (4q/e) > 4\sqrt{t}q > \log(1/\delta),
\end{align*}
so 
\begin{align*}
  \left ( \frac{et}{C^2}\right )^C < \delta.
\end{align*}
Hence we have made $\P[\text{$u$ fails} \mid v_0=u]< \delta$ by choosing $C = O(\sqrt{t}q/\log q )$.
\end{proof}

\begin{theorem}
We can simulate a $t$-step random walk on a simple undirected graph in the insertion-only model within error $\eps$ using
$O\left ( n \sqrt{t} \cdot \frac{q}{\log q} \right )$ 
words of memory, where $q = 2+\frac{\log(1/\eps)}{\sqrt{t}}$.
\end{theorem}
\begin{proof}
 The theorem follows from Lemma~\ref{union} and Lemma~\ref{capa} by setting $\delta = \frac{\eps}{2t}$.
\end{proof}

\subsection{On graphs with multiple edges}
\label{secmulti}
When the undirected graph contains multiple edges, condition~(\ref{cond}) in the proof of Lemma~\ref{capa} may not hold, so we need to slightly modify our algorithm.  

We still maintain the multiset $E_1$ of important arcs.
Whether an arc is important will be  determined by our algorithm. 
(This is different from the previous algorithm, where important arcs were simply defined as $(u,v)$ with $d(v)\le C$.)
We will ensure that condition~(\ref{cond}) still holds, i.e., for any $u\in V$ and any fixed prefix of the random walk $v_0,\dots,v_{i-1}$, 
\begin{align}
\P \big [\text{$(v_{i},v_{i-1}) \notin E_1$, and $v_{i}= u$} \,\big \vert\,  v_0,\dots,v_{i-1}\big ] < 1/C. \label{newcond}
\end{align}

Note that there can be both important arcs  and unimportant arcs from $u$ to $v$. Let $f(u,v)$ denote the number of undirected edges between $u,v$. Then there are $f(u,v)$ arcs $(u,v)$. 
Suppose $f_1(u,v)$ of these arcs are important, and $f_0(u,v) = f(u,v)-f_1(u,v)$ of them are unimportant. Then we can rewrite condition~(\ref{newcond}) as
\begin{align}
 \frac{f_0(u,v_{i-1})}{d(v_{i-1})} < 1/C, \label{newnewcond}
\end{align}
for every $u,v_{i-1} \in V$.

Similarly as before, we need to store the multiset $E_1$ using only $O(nC)$ words of space. 
And we need to sample with replacement $C$ unimportant arcs $a_{u,1},\dots,a_{u,C}$ outgoing from $u$, for every $u \in V$. 
Finally we use the procedure \textsc{SimulateRandomWalk} in Figure~\ref{fig:simulation} to simulate a random walk.

The multiset $E_1$ is determined as follows: For every vertex $v\in V$, we run Misra-Gries algorithm \cite{misra1982finding} on the sequence of  all $v$'s neighbors.
 We will obtain a list $L_v$ of at most $C$ vertices, such that for every vertex $u\notin L_v$, $\frac{f(u,v)}{d(v)} < \frac{1}{C}$. Moreover, we will get a frequency estimate $A_v(u)>0$ for every $u\in L_v$, such that $0 \le f(u,v) - A_v(u) < \frac{d(v)}{C}$.
 Assuming $A_v(u)=0$ for $u\notin L_v$,   we can satisfy condition~(\ref{newnewcond}) for all $u\in V$ by setting $f_1(u,v)=A_v(u)$.  Hence we have determined all the important arcs, and they can be stored in $O(\sum_{v}|L_v|) = O(nC)$ words. \, To sample from the unimportant arcs, we simply insert the arcs discarded by Misra-Gries algorithm into the samplers.  \, The pseudocode is given in Figure~\ref{fig:processinputmulti}.

\begin{figure}
\begin{center}
\begin{algorithmic}
\Procedure{InsertArc}{$u,v$}
\State{$d(v) \gets d(v)+1$}
\If{$u \in L_v$}
  \State{$A_v(u) \gets A_v(u)+1$}
  \Else
  \State{Insert $u$ into $L_v$}
  \State{$A_v(u) \gets 1$}
  \If{$|L_v|\ge C+1$}
  \For {$w \in L_v$}
  \State{Feed arc $(w,v)$ into $w$'s sampler}
  \State{$A_v(w) \gets A_v(w) -1$}
  \If{$A_v(w)=0$}
  \State{Remove $w$ from $L_v$}
  \EndIf
  \EndFor
  \EndIf
\EndIf
\EndProcedure

\Procedure{ProcessInput}{} 

\For{$u \in V$}
      \State{$d(u) \gets 0$} 
      \State{Initialize $u$'s sampler (initially empty) which maintains $a_{u,1},\dots,a_{u,C}$}
      \State{Initialize empty list $L_u$}
\EndFor
\For{undirected edge $(u,v)$ in the input stream} 
\State{\textsc{InsertArc}$(u,v)$}
\State{\textsc{InsertArc}$(v,u)$}
\EndFor
\State{$E_1 \gets \bigcup_{v\in V} \bigcup_{u\in L_v}\{A_v(u)\text{ copies of arc }(u,v)\}$}\Comment{Multiset of important arcs}
\EndProcedure
\end{algorithmic}
\caption{Pseudocode for processing the input stream (for undirected graphs with possibly multiple edges)} \label{fig:processinputmulti}
\end{center}
\end{figure}

 \begin{lemma}
 After  \textsc{ProcessInput} (in Figure~\ref{fig:processinputmulti}) finishes, $|L_v|\le C$. For every $u \in L_v$, $0 \le f(u,v) -A_v(u) \le \frac{d(v)}{C+1}$. For every $u\notin L_v$, $f(u,v) \le \frac{d(v)}{C+1}$. 
 \end{lemma}
 \begin{proof}
 Every time the \textbf{for} loop in procedure \textsc{InsertArc} finishes, the newly added vertex $u$ must have been removed from $L_v$, so $|L_v| \le C$ still holds. Let $W= \{w_1,\cdots,w_{C+1}\}$ be the set of  vertices in $L_v$ before this \textbf{for} loop begins. Then for every $u\in V$, $f(u,v)-A_v(u)$ equals the number of times $u$ is contained in $W$       (assuming $A_v(u) = 0$ for $u \notin L_v$), which is at most $\frac{1}{C+1}\sum_W |W| \le \frac{d(v)}{C+1}$.
 \end{proof}
 \begin{corollary}
 
 \label{coro}
  Procedure \textsc{ProcessInput} in Figure~\ref{fig:processinputmulti} computes the multiset $E_1$ of important edges and stores it using $O(nC)$ words. It also samples with replacement $C$ unimportant arcs $a_{u,1},\dots,a_{u,C}$ outgoing from $u$, for every $u \in V$. 
  Moreover, 
\begin{align*}
 \frac{f_0(u,v)}{d(v)} < \frac{1}{C} \label{newnewcond}
\end{align*}
holds for every $u,v \in V$.
 \end{corollary}

Now we analyze the failure probability of  \textsc{SimulateRandomWalk} (in Figure~\ref{fig:simulation}), similar to Lemma~\ref{capa}.
\begin{lemma}
\label{multicapa}
We can choose integer parameter 
 $C = O\left (  \sqrt{t}\cdot \frac{q}{\log q} \right ) $, where $q = 2+\frac{\log(1/\delta)}{\sqrt{t}}$, 
so that $\P[\text{$u$ fails} \mid \text{$v_0=u$}] \le \delta $ holds for every $u\in V$.
\end{lemma}

\begin{proof}
Let $d_0(u) = \sum_{v\in V} f_0(u,v)$. \ \   
As before, we rewrite this probability as a sum over possible random walks. Here we distinguish between important and unimportant arcs. Denote $s_i = \1[\text{step $(v_{i-1},v_i)$ is along an important arc}]$. 
Then for any $u \in V$, 
\begin{align*}
&\P[\text{$u$ fails} \mid \text{$v_0=u$}]\\
\le{} &\P[\text{$u$ fails} \mid \text{$v_0=u$, arc $(v_0,v_1)$ is unimportant}]\\
={}&\frac{d(u)}{d_0(u)}\sum_{k\le t}\sum_{(v_0,\dots,v_k)}\sum_{s_1,\dots,s_k}\1\bigg [\text{$v_0=v_{k-1}=u,\, s_1=s_k=0,$}\nonumber\\[-6\jot]
&\hspace{5cm}\text{$ |\{i: v_i = u, s_{i+1}=0\}|=C+1$} \bigg ]\prod_{i=0}^{k-1}\frac{f_{s_{i+1}}(v_i,v_{i+1})}{d(v_i)}\\[+1\jot]
={}&\sum_{k\le t}\sum_{(v_0,\dots,v_{k-1})}\sum_{s_1,\dots,s_{k-1}}\1\bigg [\text{$v_0=v_{k-1}=u,\,\, s_1=0,$}\nonumber\\[-6\jot]
&\hspace{5cm}\text{$|\{i: v_i = u, s_{i+1}=0\}|=C$}\bigg ]\prod_{i=0}^{k-2}\frac{f_{s_{i+1}}(v_i,v_{i+1})}{d(v_i)}.
 \end{align*}
 Let $v'_i = v_{k-1-i}, s'_i = s_{k-i}$. Then this sum  equals
\begin{align*}
& \sum_{k\le t}\sum_{(v'_0,\dots,v'_{k-1})} \sum_{s'_1,\dots,s'_{k-1}}\1\bigg [\text{$v'_0=v'_{k-1}=u,\,\, s'_{k-1}=0,$}\nonumber\\[-6\jot]
&\hspace{6cm}\text{$|\{i: s'_i =0 , v'_{i}=u\}|=C$}\bigg ]\prod_{i=1}^{k-1}\frac{f_{s'_{i}}(v'_{i},v'_{i-1})}{d(v'_{i-1})}\\[+3\jot]
={}&\underset{\text{random walk $(v'_0,\dots,v'_{t-1}) $}}{\P}\Big [\,|\{i: \text{$v'_{i}=u$, arc $(v'_{i},v'_{i-1})$ is unimportant}\}|\ge C \, \Big \vert \, v'_0=u\Big ].
\end{align*}
Notice that for any $i$ and any fixed prefix $v'_0,\dots,v'_{i-1}$,
\begin{align*}
 \P \Big [\text{$v'_{i}=u$, arc $(v'_{i},v'_{i-1})$ is unimportant}\, \Big \vert \, v'_0,v'_1,\dots,v'_{i-1}\Big ] = \frac{f_0(u,v'_{i-1})}{d(v'_{i-1})} <  \frac{1}{C}
\end{align*}
by Corollary~\ref{coro}. The rest of the proof is the same as in Lemma~\ref{capa}.
\end{proof}
\begin{theorem}
\label{mainthmmulti}
We can simulate a random walk on an undirected graph with possibly multiple edges in the insertion-only model within error $\eps$ using $O\left ( n \sqrt{t} \cdot \frac{q}{\log q} \right )$ 
words of memory, where $q = 2+\frac{\log(1/\eps)}{\sqrt{t}}$.   
\end{theorem}
\begin{proof}
 The theorem follows from Lemma~\ref{union} and Lemma~\ref{multicapa} by setting $\delta = \frac{\eps}{2t}$.
\end{proof}
\section{Turnstile model}
In this section we consider the turnstile model where both insertion and deletion of edges can appear. 
\begin{lemma}[$\ell_1$ sampler in the turnstile model, \cite{jayaram2018perfect}]
Let $f \in \R^n$ be a vector defined by a stream of updates to its coordinates of the form $f_i \gets f_i + \Delta$, where $\Delta$ can either be positive or negative. There is an algorithm which reads the stream and returns an index $i \in [n]$ such that for every $j \in [n]$, 
\begin{equation}
\label{sampleprob}
 \P[i=j] = \frac{|f_j|}{\|f\|_1} + O(n^{-c}), 
\end{equation}
where $c\ge 1$ is some arbitrary large constant. It is allowed to output \textsc{Fail}  with probability $\delta$, and in this case it will not output any index. The space complexity of this algorithm is $O(\log^2 n \log (1/\delta))$ bits.
\end{lemma}

\begin{remark}
For $\eps \ll 1/n$, the $O(n^{-c})$ error term in (\ref{sampleprob}) can be reduced to $O(\eps^c)$ by running the $\ell_1$ sampler on $f \in \R^{\lceil 1/\eps \rceil}$, using $O(\log^2 (1/\eps) \log(1/\delta))$ bits of space.
\end{remark}

We will use the $\ell_1$ sampler for sampling neighbors (with possibly multiple edges) in the turnstile model. The error term $O(n^{-c})$ (or $O(\eps^c)$) in (\ref{sampleprob}) can be ignored in the following discussion, by choosing sufficiently large constant $c$ and scaling down $\eps$ by a constant.

\subsection{Directed graphs}

\begin{theorem}
\label{turnstilealgdirected}
  We can simulate a $t$-step random walk on a directed graph in the turnstile model within error $\eps$ using $O(n(t+ \log\frac{1}{\eps})\log^2 \max\{n,1/\eps\})$ bits of memory.
\end{theorem}
\begin{proof}
For every $u \in V$, we run $C' = 2t + 16\log (2t/\eps)$ independent $\ell_1$ samplers each having failure probability $\delta=1/2$. We use them to sample the outgoing edges of $u$ (as in the algorithm of Theorem~\ref{algdirected}). By Chernoff bound, the probability that less than $t$ samplers succeed is at most $\eps/(2t)$.

We say a vertex $u$ fails if $u$ has less than $t$ successful samplers, and  $u \in \{v_0,v_1,\dots,v_{t-1}\}$  (where $v_0,v_1,\dots,v_t$ is the random walk). Then $\P[\text{$u$ fails}] \le \frac{\eps}{2t} \P[u \in \{v_0,\dots,v_{t-1}\}]$. By union bound, $\P[\text{at least one vertex fails}] \le \frac{\eps}{2t}\sum_{u\in V} \P[u \in \{v_0,\dots,v_{t-1}\}] \le \frac{\eps}{2}$. Hence, with probability $1-\frac{\eps}{2}$, every vertex $u$ visited (except the last one) has at least $t$ outgoing edges sampled, so our simulation can succeed. The space usage is $O(nC' \log^2 \max\{n,1/\eps\} \log (1/\delta)) = O(n(t+ \log\frac{1}{\eps}) \log^2 \max\{n,1/\eps \})$ bits.
 \end{proof}

\subsection{Undirected graphs}
We slightly modify the \textsc{ProcessInput} procedure of our previous algorithm in Section~\ref{secmulti}. We will use the $\ell_1$ heavy hitter algorithm in the turnstile model.
\begin{lemma}[$\ell_1$ heavy hitter, \cite{cormode2005improved}]
 Let $f \in \R^n$ be a vector defined by a stream of updates to its coordinates of the form $f_i \gets f_i + \Delta$, where $\Delta$ can either be positive or negative. There is a randomized algorithm which reads the stream and returns a subset $L \subseteq [n]$ 
 such that $i \in L$ for every $|f_i| \ge \frac{\|f\|_1}{k}$, and $i \notin L$ for every $|f_i| \le \frac{\|f\|_1}{2k}$. Moreover it returns a frequency estimate $\tilde f_i$ for every $i \in L$, which satisfies $0 \le f_i  -\tilde f_i \le \frac{\|f\|_1}{2k}$.  The failure probability of this algorithm is $O(n^{-c})$. The space complexity is $O(k\log^2 n)$ bits.
\end{lemma}
\begin{remark}
For $\eps \ll 1/n$, the $O(n^{-c})$ failure probability of this $\ell_1$ heavy hitter algorithm can be reduced to $O(\eps^c)$ by running the algorithm on $f \in \R^{\lceil 1/\eps \rceil}$, using $O(k \log^2 (1/\eps))$ bits of space. In the following discussion, this failure probability can be ignored by making the constant $c$ sufficiently large.
\end{remark}

\begin{theorem}
\label{turnstilealgundirected}
  We can simulate a $t$-step random walk on an undirected graph in the turnstile model within error $\eps$ using $O(n(\sqrt{t}+ \log\frac{1}{\eps})\log^2 \max\{n,1/\eps\})$ bits of memory.
\end{theorem}

\begin{proof}
Similar to the previous insertion-only algorithm (in Figure~\ref{fig:processinputmulti}), we perform two \textit{arc updates} $((u,v),\Delta),\, ((v,u),\Delta)$ when we read an \textit{edge update} $((u,v),\Delta)$ from the stream.

For every $u \in V$, we run $C' = 2C + 16\log (2t/\eps)$ independent $\ell_1$ samplers each having failure probability $\delta=1/2$, where $C$ is the same constant as in the proof of Lemma~\ref{multicapa} and Theorem~\ref{mainthmmulti}. By Chernoff bound, the probability that less than $C$ samplers succeed is at most $\eps/(2t)$. For every arc update $((u,v),\Delta)$, we send update ($v$, $\Delta$) to $u$'s $\ell_1$ sampler.

In addition, for every $v\in V$, we run $\ell_1$ heavy hitter algorithm with $k= C$. For every arc update $((u,v),\Delta)$, we send update $(u,\Delta)$ to  $v$'s heavy hitter algorithm. In the end, we will get a frequency estimate $A_v(u)$  for every $u \in V$, such that $f(u,v) - \frac{d(v)}{C}\le A_v(u) \le f(u,v)$. We then insert $A_v(u)$ copies of arc $(u,v)$ into $E_1$  (the multiset of important arcs), and send update $(v, -A_v(u))$ to $u$'s $\ell_1$ sampler. Then we use the $\ell_1$ samplers to sample unimportant arcs for every $u$.

As before, we use the procedure \textsc{SimulateRandomWalk} (in Figure~\ref{fig:simulation}) to simulate the random walk. The analysis of the failure probability of the $\ell_1$ samplers is the same as in Theorem~\ref{turnstilealgdirected}. The analysis of the failure probability of procedure \textsc{SimulateRandomWalk} is the same as in Lemma~\ref{multicapa}. The space usage of the algorithm is $O(nC' \log^2 \max\{n,1/\eps\} \log \delta) = O(n(\sqrt{t}+\log \frac{1}{\eps}) \log^2 \max\{n,1/\eps\})$ bits.
\end{proof}

\section{Conclusion}
We end our paper by discussing some related questions for future research.

\begin{itemize}
\item The output distribution of our insertion-only algorithm for undirected graphs is $\eps$-close to the random walk distribution. 
What if the output is required to be perfectly random, i.e., $\eps=0$? 

\item For insertion-only simple undirected graphs, we proved an $\Omega(n\sqrt{t})$-bit space lower bound. Our algorithm uses $O(n\sqrt{t} \log n)$ bits (for not too small $\eps$). Can we close the gap between the lower bound and the upper bound, as in the case of directed graphs? 

\item In the undirected version, suppose the starting vertex $v_0$ is drawn from a distribution (for example, the stationary distribution of the graph) rather than being specified.
Is it possible to obtain a better algorithm in this new setting? Notice that our proof of the $\Omega(n\sqrt{t})$ lower bound does not work here, since it requires $v_0$ to be specified.

\item We required the algorithm to output all vertices on the random walk.
If only the last vertex is required, can we get a better algorithm or prove non-trivial lower bounds?
\end{itemize}
%%
%% Bibliography
%%

%% Please use bibtex, 
%%
%% Bibliography
%%

%% Please use bibtex, 

\bibliography{lipics-v2018-sample-article}

\begin{thebibliography}{10}

\bibitem{ahn2013linear}
Kook~Jin Ahn and Sudipto Guha.
\newblock Linear programming in the semi-streaming model with application to
  the maximum matching problem.
\newblock {\em Information and Computation}, 222:59--79, 2013.
\newblock \href {http://dx.doi.org/10.1016/j.ic.2012.10.006}
  {\path{doi:10.1016/j.ic.2012.10.006}}.

\bibitem{ahn2012analyzing}
Kook~Jin Ahn, Sudipto Guha, and Andrew McGregor.
\newblock Analyzing graph structure via linear measurements.
\newblock In {\em Proceedings of the 23rd Annual ACM-SIAM Symposium on Discrete
  Algorithms (SODA)}, pages 459--467, 2012.
\newblock \href {http://dx.doi.org/10.1137/1.9781611973099.40}
  {\path{doi:10.1137/1.9781611973099.40}}.

\bibitem{andersen2007using}
Reid Andersen, Fan Chung, and Kevin Lang.
\newblock Using pagerank to locally partition a graph.
\newblock {\em Internet Mathematics}, 4(1):35--64, 2007.
\newblock \href {http://dx.doi.org/10.1080/15427951.2007.10129139}
  {\path{doi:10.1080/15427951.2007.10129139}}.

\bibitem{andersen2009finding}
Reid Andersen and Yuval Peres.
\newblock Finding sparse cuts locally using evolving sets.
\newblock In {\em Proceedings of the 41st Annual ACM Symposium on Theory of
  Computing (STOC)}, pages 235--244, 2009.
\newblock \href {http://dx.doi.org/10.1145/1536414.1536449}
  {\path{doi:10.1145/1536414.1536449}}.

\bibitem{charikar2003better}
Moses Charikar, Liadan O'Callaghan, and Rina Panigrahy.
\newblock Better streaming algorithms for clustering problems.
\newblock In {\em Proceedings of the 35th Annual ACM Symposium on Theory of
  Computing (STOC)}, pages 30--39, 2003.
\newblock \href {http://dx.doi.org/10.1145/780542.780548}
  {\path{doi:10.1145/780542.780548}}.

\bibitem{cormode2005improved}
Graham Cormode and Shan Muthukrishnan.
\newblock An improved data stream summary: the count-min sketch and its
  applications.
\newblock {\em Journal of Algorithms}, 55(1):58--75, 2005.
\newblock \href {http://dx.doi.org/10.1016/j.jalgor.2003.12.001}
  {\path{doi:10.1016/j.jalgor.2003.12.001}}.

\bibitem{sarma2011estimating}
Atish Das~Sarma, Sreenivas Gollapudi, and Rina Panigrahy.
\newblock Estimating pagerank on graph streams.
\newblock {\em Journal of the ACM (JACM)}, 58(3):13, 2011.
\newblock \href {http://dx.doi.org/10.1145/1970392.1970397}
  {\path{doi:10.1145/1970392.1970397}}.

\bibitem{epstein2011improved}
Leah Epstein, Asaf Levin, Juli{\'a}n Mestre, and Danny Segev.
\newblock Improved approximation guarantees for weighted matching in the
  semi-streaming model.
\newblock {\em SIAM Journal on Discrete Mathematics}, 25(3):1251--1265, 2011.
\newblock \href {http://dx.doi.org/10.1137/100801901}
  {\path{doi:10.1137/100801901}}.

\bibitem{jayaram2018perfect}
Rajesh Jayaram and David~P. Woodruff.
\newblock Perfect lp sampling in a data stream.
\newblock In {\em Proceedings of the 59th Annual IEEE Symposium on Foundations
  of Computer Science (FOCS)}, pages 544 -- 555, 2018.
\newblock \href {http://dx.doi.org/10.1109/FOCS.2018.00058}
  {\path{doi:10.1109/FOCS.2018.00058}}.

\bibitem{jerrum1989approximating}
Mark Jerrum and Alistair Sinclair.
\newblock Approximating the permanent.
\newblock {\em SIAM Journal on Computing}, 18(6):1149--1178, 1989.
\newblock \href {http://dx.doi.org/10.1137/0218077}
  {\path{doi:10.1137/0218077}}.

\bibitem{jerrum1986random}
Mark~R. Jerrum, Leslie~G. Valiant, and Vijay~V. Vazirani.
\newblock Random generation of combinatorial structures from a uniform
  distribution.
\newblock {\em Theoretical Computer Science}, 43:169--188, 1986.
\newblock \href {http://dx.doi.org/10.1016/0304-3975(86)90174-X}
  {\path{doi:10.1016/0304-3975(86)90174-X}}.

\bibitem{kapralov2013better}
Michael Kapralov.
\newblock Better bounds for matchings in the streaming model.
\newblock In {\em Proceedings of the 24th Annual ACM-SIAM Symposium on Discrete
  Algorithms (SODA)}, pages 1679--1697, 2013.
\newblock \href {http://dx.doi.org/10.1137/1.9781611973105.121}
  {\path{doi:10.1137/1.9781611973105.121}}.

\bibitem{kapralov2017single}
Michael Kapralov, Yin~Tat Lee, Cameron Musco, Christopher Musco, and Aaron
  Sidford.
\newblock Single pass spectral sparsification in dynamic streams.
\newblock {\em SIAM Journal on Computing}, 46(1):456--477, 2017.
\newblock \href {http://dx.doi.org/10.1137/141002281}
  {\path{doi:10.1137/141002281}}.

\bibitem{kelner2013spectral}
Jonathan~A. Kelner and Alex Levin.
\newblock Spectral sparsification in the semi-streaming setting.
\newblock {\em Theory of Computing Systems}, 53(2):243--262, 2013.
\newblock \href {http://dx.doi.org/10.1007/s00224-012-9396-1}
  {\path{doi:10.1007/s00224-012-9396-1}}.

\bibitem{miltersen1998data}
Peter~Bro Miltersen, Noam Nisan, Shmuel Safra, and Avi Wigderson.
\newblock On data structures and asymmetric communication complexity.
\newblock {\em Journal of Computer and System Sciences}, 57(1):37 -- 49, 1998.
\newblock \href {http://dx.doi.org/10.1006/jcss.1998.1577}
  {\path{doi:10.1006/jcss.1998.1577}}.

\bibitem{misra1982finding}
J.~Misra and David Gries.
\newblock Finding repeated elements.
\newblock {\em Science of Computer Programming}, 2(2):143 -- 152, 1982.
\newblock \href {http://dx.doi.org/10.1016/0167-6423(82)90012-0}
  {\path{doi:10.1016/0167-6423(82)90012-0}}.

\bibitem{reingold2008undirected}
Omer Reingold.
\newblock Undirected connectivity in log-space.
\newblock {\em Journal of the ACM (JACM)}, 55(4):17, 2008.
\newblock \href {http://dx.doi.org/10.1145/1391289.1391291}
  {\path{doi:10.1145/1391289.1391291}}.

\bibitem{spielman2013local}
Daniel~A. Spielman and Shang-Hua Teng.
\newblock A local clustering algorithm for massive graphs and its application
  to nearly linear time graph partitioning.
\newblock {\em SIAM Journal on Computing}, 42(1):1--26, 2013.
\newblock \href {http://dx.doi.org/10.1137/080744888}
  {\path{doi:10.1137/080744888}}.

\end{thebibliography}

\end{document}